\documentclass[letterpaper, 10 pt, conference]{ieeeconf}  

\usepackage{lipsum}
\usepackage{amssymb}
\usepackage{amsmath}
\usepackage[pdftex]{graphicx}
\usepackage{color}

\usepackage{sgame}

\usepackage{comment}
\usepackage{algorithm}
\usepackage{algorithmic}

\usepackage{multicol}

\newtheorem{theorem}{Theorem}

\newtheorem{proposition}{Proposition}

\newtheorem{definition}{Definition}

\newtheorem{lemma}{Lemma}

\newtheorem{remark}{Remark}
\newtheorem{example}{Example}

\usepackage{cite}

\DeclareMathOperator{\Equaldef}{\overset{def}{=}}

\usepackage[capitalize]{cleveref}

\title{\LARGE \bf
Strategic multi-task coordination over regular networks of robots with limited computation and communication capabilities}

\author{Yi Wei and Marcos M. Vasconcelos 
\thanks{Yi Wei and Marcos M. Vasconcelos are with the Commonwealth Cyber Initiative and the Bradley Department of Electrical and Computer Engineering,
        Virginia Polytechnic Institute and State University, Arlington, VA 22203, USA. E-mail: 
        {\tt\small \{wyi,marcosv\}@vt.edu}}%
}

\begin{document}

\maketitle
\thispagestyle{empty}
\pagestyle{empty}

\begin{abstract}


Coordination is a desirable feature in multi-agent systems, allowing the execution of tasks that would be impossible by individual agents. We study coordination by a team of strategic agents choosing to undertake one of the multiple tasks. We adopt a stochastic framework where the agents decide between two distinct tasks whose difficulty is randomly distributed and partially observed. We show that a Nash equilibrium with a simple and intuitive linear structure exists for diffuse prior distributions on the task difficulties. Additionally, we show that the best response of any agent to an affine strategy profile can be nonlinear when the prior distribution is not diffuse. Finally, we state an algorithm that allows us to efficiently compute a data-driven Nash equilibrium within the class of affine policies.


\end{abstract}

\section{Introduction}

Task allocation among the members of a group of agents is a classic problem in robotics with a rich history and many applications. Traditionally, task allocation consists of assigning each agent to one of many possible tasks such that a given performance metric is optimized. However, depending on the structure of the objective function, this problem can be computationally intensive. Furthermore, its decentralized implementation may require local communication among the agents, which may be challenging to implement in large-scale systems due to interference. It is also highly costly if the agents operate under limited batteries. An alternative to this approach is to allow the agents to sense the tasks and let them make autonomous decisions about which job to undertake based on a strategy.

We consider a framework for multi-task allocation inspired by a class of Bayesian games known as Global Games \cite{Carlsson:1993,Morris:2003}. In this class of games, the system's state is partially observed by the agents, who do not communicate among themselves, but whose actions affect other agents within their neighborhood. 
The power of such a simple framework is that under simple assumptions on the structure of the payoff function of the agents in the system, structured Bayesian Nash equilibria exist. Moreover, this framework often resolves the problem of selecting one of the multiple equilibria when the system's state is perfectly observed.

The system model herein is depicted in \cref{fig:system}, where a finite number of agents must choose between one of two tasks. Each task is characterized by independent random variables corresponding to their difficulty level, which we call the system's state. The tasks are distributed uniformly among all agents in the system.
However, the coordination among the agents required to complete the overall task occurs by means of local interactions over a network. Additionally, each agent has partial knowledge about the state consisting of two noisy signals, one for each task. Finally, The agent then decides on which of the tasks to undertake, and its associated payoff depends on the chosen task and on the number of neighbors that have made the same decision. We are interested in obtaining strategies that constitute a Bayesian Nash Equilibrium (BNE). 

\begin{figure}[t!]
    \centering
    \includegraphics[width=\columnwidth]{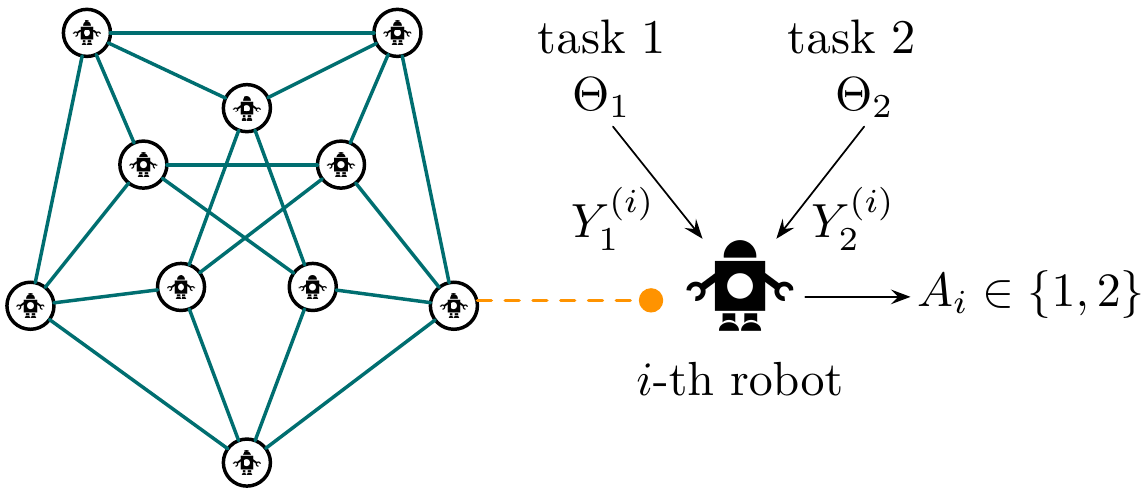}
    \caption{Block diagram for the Multi-Task Global Game considered in this paper.}
    \label{fig:system}
\end{figure}

Game theoretic tools that enable coordination have been used to study task allocation for networks of autonomous agents. Saad et al. \cite{Saad:2010} used a hedonic coalition formation game for data collection tasks. Jang et al. \cite{Jang:2018}  studied the system in \cite{Saad:2010} when the number of agents is very large. 
Krishnamurthy \cite{Krishnamurthy:2009} used global games to coordinate multiple  channel access by a large number of cognitive radios.


This model is inspired by the work in \cite{Kanakia:2016}, which considers a single task, and the agent's decision corresponds to undertaking a task or not. The main result in \cite{Kanakia:2016} is that, under certain structural assumptions of the agents' payoffs, there exists a BNE where every agent uses a threshold strategy. When there are multiple tasks the characterization of BNE solutions is much more involved. We show that in the limit when the prior distributions of the state variables is \textit{diffuse}, i.e., the variance is asymptotically large, there is a BNE where each agent takes that task corresponding to the smallest measurement. When the priors distributions are not diffuse, we obtain sufficient conditions such that the best-response to a set of affine strategies is characterized by a monotone nonlinear switching function. Then we propose an algorithm that searches for a BNE within the class of affine strategies. We provide numerical results which show how the level of coordination degrades with respect to the variance of the additive noise in the sensing channel. 










\section{System Model}

Consider a system with $N > 1$ agents indexed by $i\in\{1,\cdots,N\}=[N]$. The agents are autonomous decision makers and chose one out of two tasks which they would like to work on. When an agent undertakes a Task $t$, it contributes a unit of effort in completing it. 
 Let the difficulty of completing the Task $t$ be randomly selected from a zero mean Gaussian prior distribution on $\mathbb{R}$, i.e.,
\begin{equation}
    \Theta_t \sim 
    \mathcal{N}(0,\sigma_t^2),\ \ t\in\{1,2\}.
\end{equation}

The variables $\Theta_1$ and $\Theta_2$ are independent. Each agent observes two noisy signals about the task difficulties. Let $Y^{(i)}_t$ denote the signal received by the $i$-th agent about the $t$-th fundamental, and be given as follows:
\begin{equation}
    Y^{(i)}_t = \Theta_t +Z^{(i)}_t,
\end{equation}
where $\{Z^{(i)}_1\}_{i=1}^N$ and $\{Z^{(i)}_2\}_{i=1}^N$ are independent and identically distributed sequences of random variables such that
\begin{equation}
Z^{(i)}_t    \sim \mathcal{N}(0,\alpha^2_{t}), \ \ i\in[N].
\end{equation}
The sequences $\{Z^{(i)}_1\}_{i=1}^N$ and $\{Z^{(i)}_2\}_{i=1}^N$ are mutually independent. 

Each agent makes a binary decision $A_i\in \{1,2\}$, which corresponds to an agent deciding to undertake Task 1 or Task 2, respectively. The $i$-th agent uses a measurable policy $\gamma_i : \mathbb{R}^2 \rightarrow \{1,2\}$ on its observations to compute its decision. The policy is a mapping from the observation space of the noisy signal to the decision to tackle Task 1 or 2. 
Let $Y^{(i)}=[Y^{(i)}_1,Y^{(i)}_2]^{\mathsf{T}}$, the action of the $i$-th agent is computed according to:
\begin{equation}
    A_i = \gamma_i(Y^{(i)}).    
\end{equation}

\subsection{Bayesian Coordination Games}

Assume that the agents interact according to an influence graph $\mathcal{G}=([N],\mathcal{E})$. We assume that the graph is connected, undirected and with unit weights. Let $\mathcal{N}_i$ denote the set of neighbors of agent $i \in [N]$.

Consider two agents $i,j \in [N]$ such that $(i,j)\in\mathcal{E}$. Because these agents are connected, we assume that they engage in a \textit{coordination game} defined by the  payoff matrix in \cref{fig:bimatrix}.
The idea behind the payoff matrix is that Task $t$ with difficulty $\theta_t$ becomes easier as the number agents in the system increases. However, a task will be successfully completed only if a sufficient fraction of agents within an agents neighborhood collaborates in the same task. 

\begin{figure}[b]\hspace*{\fill}%
\centering
\begin{game}{2}{2}[$a_i$][$a_j$]
     & $1$ & $2$             \\
 $1$ & $\Big(\frac{1}{|\mathcal{N}_i|}-\frac{\theta_1}{N},\frac{1}{|\mathcal{N}_j|}-\frac{\theta_1}{N}\Big)$ & $\Big(-\frac{\theta_1}{N},-\frac{\theta_2}{N}\Big)$ \\
 $2$ & $\Big(-\frac{\theta_2}{N},-\frac{\theta_1}{N}\Big)$ & $\Big(\frac{1}{|\mathcal{N}_i|}-\frac{\theta_2}{N},\frac{1}{|\mathcal{N}_j|}-\frac{\theta_2}{N}\Big)$  \\
\end{game}\hspace*{\fill}%
\caption[]{A binary action coordination game between two players deciding between two tasks with difficulties $\theta_1$ and $\theta_2$, respectively.}
\label{fig:bimatrix}
\end{figure}

\vspace{5pt}

\begin{proposition}
Consider the bimatrix game in \cref{fig:bimatrix}. Let $L_{ij}\Equaldef N/\min\{|\mathcal{N}_i|,|\mathcal{N}_j|\}$ and $\mathcal{S}_{ij}$ denote the set of pure-strategy Nash-Equilibria of the deterministic coordination game of \cref{fig:bimatrix}. The following holds:
\begin{align}
\theta_1-\theta_2 &<-L_{ij} \Rightarrow \mathcal{S}_{ij} = \{(1,1)\}; \\
-L_{ij}\leq\theta_1-\theta_2&\leq +L_{ij} \Rightarrow \mathcal{S}_{ij} = \{(1,1),(2,2)\}; \\ 
\theta_1-\theta_2 &> +L_{ij} \Rightarrow \mathcal{S}_{ij} = \{(2,2)\}.
\end{align}
\end{proposition}

\vspace{5pt}

\begin{proof}
The proof can be obtained by inspection using the definition of a Nash-Equilibrium \cite{Fudenberg:1998}.
\end{proof}

\vspace{5pt}

Therefore, there exists a regime in which the agents $i,j$ can choose to play one of multiple equilibria depending on the values of the variables $\theta_1$ and $\theta_2.$ Moreover, the set of equilibria depends on the neighborhood structure of the graph. 

Let $\theta\Equaldef [\theta_1 \ \theta_2]^{\mathsf{T}}$. The utility of agent $i$ in the game with one of its neighbors $j$ is given by:
\begin{multline}
\mathcal{U}_i(a_i,a_j,\theta) \Equaldef \mathbf{1}(a_i=1) \Big( \frac{1}{|\mathcal{N}_i|}\mathbf{1}(a_j=1)-\frac{\theta_1}{N}\Big)  \\ +\mathbf{1}(a_i=2) \Big( \frac{1}{|\mathcal{N}_j|}\mathbf{1}(a_j=2)-\frac{\theta_2}{N}\Big).
\end{multline}
In a graphical coordination game, agent $i$ simultaneously plays a single action with all of its neighbors $j$ in $\mathcal{N}_i$ \cite{Paarporn:2021,Arditti:2021}. Therefore, we define the following local utility function as
\begin{equation}
\mathcal{V}_i(a_i,a_{-i},\theta) \Equaldef \sum_{j\in\mathcal{N}_i} \mathcal{U}_i(a_i,a_j,\theta).
\end{equation}

\subsection{Bayesian Nash Equilibrium}

The agents in in this game act in a noncooperative manner by seeking to maximize their individual expected payoff with respect to their strategy policies \cite{Fudenberg:1998}. Let $\Gamma$ be a \textit{coordination strategy profile}, i.e., the collection of coordination policies used by all the agents in the system. Define 
\begin{equation}
\Gamma \Equaldef (\gamma_1, \cdots, \gamma_N),
\end{equation}
where $\gamma_i: \mathbb{R}^2\rightarrow \{1,2\}$ is a measurable function, $i\in[N]$. Similarly, let $\Gamma_{-i}$ denote the coordination strategy profile used by all the agents except agent $i$.

Given $\Gamma_{-i}$, the goal of the $i$-th agent is to solve the following (stochastic) optimization problem: 
\begin{equation}
    \underset{\gamma_i}{\mathrm{maximize}} \ \  \mathcal{J}_i(\gamma_i,\Gamma_{-i}) \Equaldef  \mathbf{E} \Big[ \mathcal{V}_i(A_i,A_{-i},\Theta) \Big],
\end{equation}
where the expectation is taken over the exogenous random variables $\Theta$, and $Z^i_t$, $i\in [N]$, $t\in\{1,2\}$. We are interested in obtaining BNE strategies.

\vspace{5pt}

\begin{definition}[Bayesian Nash Equilibrium]
The strategy profile $\mathbb{\gamma}^\star = (\gamma_1^\star,\cdots,\gamma_N^\star)$ is a Bayesian Nash Equilibrium for the multi-task Gaussian global game if:
\begin{equation}
\mathcal{J}_i(\gamma_i^\star,\gamma_{-i}^\star) \geq 
\mathcal{J}_i(\gamma_i,\gamma_{-i}^\star), \ \ \gamma_i \in \Gamma_i, \ \ \forall i \in [N].
\end{equation}
\end{definition}

\vspace{5pt}

\section{Definitions and Main results}

\subsection{Affine coordination strategies}

In a traditional Global Game formulations with a single task, under certain conditions, there is a unique BNE in the class of threshold policies, in which each agent simply compares its observation to a threshold \cite{Kanakia:2016,Carlsson:1993,Morris:2003,Mahdavifar:2017,Vasconcelos:2018,Vasconcelos:2022}. However, in our multi-task game, the class of policies of interest is a generalized version of a threshold policy in the following sense: each agent computes a linear function of their noisy observations and compare the result to a threshold. Therefore, the overall coordination policy is an affine function of the observation vector. 

\vspace{5pt}

\begin{definition}[Affine coordination strategy]
Let the observations of the $i$-th agent be denoted by $y^{(i)} \in \mathbb{R}^2$. An affine strategy for the multi-task global game is a function with the following structure:
\begin{equation}
    \gamma_i(y^{(i)}) = \begin{cases}
    1, &  \text{if} \ \ \lambda^{(i)\mathsf{T}}y^{(i)}\leq \tau^{(i)} \\
    2, & \text{otherwise,}
    \end{cases}
\end{equation}
where $\lambda^{(i)} \in \mathbb{R}^2$ and $\tau^{(i)} \in \mathbb{R}$.

\vspace{5pt}


\begin{remark}
A similar strategy structure has been considered in \cite{Krishnamurthy:2009}. However, the game studied in \cite{Krishnamurthy:2009} is very different from ours. The main difference is that in our setup, the agents start the game without a preassigned task, and must choose between two tasks to work on. Whereas in \cite{Krishnamurthy:2009} the agents start the game at a task and must choose between either switching to a different task or continuing at the current one.  
\end{remark}

\end{definition}

\subsection{Main result}


\begin{theorem}
Consider the multi-task global game with Gaussian observations. Consider the strategy profile $\gamma^{\min} = (\gamma^{\min}_1,\cdots, \gamma^{\min}_N)$, where 
\begin{equation}\label{eq:min_policy}
    \gamma^{\min}_i(y^{(i)}) = \begin{cases}
    1, &  \text{if} \ \ y^{(i)}_1 \leq y^{(i)}_2  \\
    2, & \text{otherwise.}
    \end{cases}
\end{equation}
If $\sigma_t^2\rightarrow \infty$, $t\in\{1,2\}$, then $\gamma^{\min}$ is a Bayesian Nash Equilibrium.
\end{theorem}

\vspace{5pt}

\section{Analysis}

Our results hinge on the analysis of the best response strategy of an agent given its observations about the state of system. Consider an agent $i\in[N]$ and fix the strategies of every other agent in the system $\Gamma_{-i}$. The best response of agent $i$ is determined by the following inequality:
\begin{multline}
\mathbf{E}\Bigg[ \frac{1}{|\mathcal{N}_i|}\sum_{j\in\mathcal{N}_i}\mathbf{1}(A_j=1) - |\mathcal{N}_i|\frac{\Theta_1}{N} \ \bigg| \ Y^{(i)}=y^{(i)} \Bigg] \stackrel{a_i^\star=1}{\geq} \\ \mathbf{E}\Bigg[ \frac{1}{|\mathcal{N}_i|}\sum_{j\in\mathcal{N}_i}\mathbf{1}(A_j=2) - |\mathcal{N}_i|\frac{\Theta_2}{N} \mid Y^{(i)}=y^{(i)} \Bigg], 
\end{multline}
which, after some algebraic manipulations, can be written as:
\begin{multline}
\sum_{j\in\mathcal{N}_i} \mathbf{P}(A_j=1\mid Y^{(i)}=y^{(i)}) \stackrel{a_i^\star=1}{\geq} \\ \frac{|\mathcal{N}_i|}{2} + \frac{|\mathcal{N}_i|^2}{2N}\mathbf{E}\Big[ \Theta_1 - \Theta_2 \ \Big| \ Y^{(i)}=y^{(i)}\Big].
\end{multline}
The decision rule above has a very intuitive interpretation: agent $i$ will choose Task $1$ if its aggregate belief on the actions of its neighbors is larger than half of the size of its neighborhood adjusted by a bias term proportional to the conditional expectation (minimum mean square estimate) of the difference in task difficulty. For instance, after observing its local signal $Y^{(i)}=y^{(i)}$, the task difficulties have the same mean, the best response for agent $i$ is equivalent to a simple majority rule. 

\vspace{5pt}

\vspace{5pt}

\subsection{Belief on the actions of agents using affine strategies}

Assume that agent $j$ uses an affine coordination strategy parametrized by $\lambda^{(j)}$ and $\tau^{(j)}$. Let the belief on the action of agent $j$ by agent $i$ given its observed vector $y^{(i)}$ be denoted by $\pi_{ij}(y^{(i)})$. Let $\Phi(\cdot)$ denote the cumulative distribution function of a standard Gaussian random variable. The following identity holds:
\begin{multline}
\pi_{ij}(y^{(i)})  =   \int_{\mathbb{R}^2}\Phi\left(\frac{\tau^{(j)}-\lambda^{(j)T}\theta}{\sqrt{(\lambda^{(j)}_1\alpha_1)^2+(\lambda^{(j)}_2\alpha_2)^2}} \right)\\\times f_{\Theta\mid Y^{(i)}=y^{(i)}}(\theta) d\theta,
\end{multline}
where 
\begin{equation}
f_{\Theta\mid Y^{(i)}=y^{(i)}}(\theta) =f_{\Theta_1\mid Y_1^{(i)}=y_1^{(i)}}(\theta_1)f_{\Theta_2\mid Y_2^{(i)}=y_2^{(i)}}(\theta_2),
\end{equation}
with
\begin{equation}
\Theta_t\mid Y^{(i)}_{t}=y^{(i)}_{t} \sim \mathcal{N}\left(\frac{\sigma_t^2}{\sigma_t^2+\alpha_t^2}\cdot y^{(i)}_{t},\frac{\alpha_t^2\sigma_t^2}{\sigma_t^2+\alpha_t^2}\right).
\end{equation}

\vspace{5pt}

Using the the fact that linear combination of Gaussian random variables is Gaussian, we can express the belief as:
\begin{equation}
\pi_{ij} (y_i) = \mathbf{E}\left[\Phi\left( \frac{\tau_j - d_1\lambda^{(j)}_1y^{(i)}_1 - d_2\lambda^{(j)}_2y^{(i)}_2}{\sqrt{(\lambda^{(j)}_1\alpha_1)^2+(\lambda^{(j)}_2\alpha_2)^2}}
- W_j\right) \right],
\end{equation}
where
\begin{equation}
W_j\sim\mathcal{N}\left(0, \frac{(\lambda^{(j)}_1\tilde{\sigma}_1)^2+(\lambda^{(j)}_2\tilde{\sigma}_2)^2}{(\lambda^{(j)}_1\alpha_1)^2+(\lambda^{(j)}_2\alpha_2)^2}\right)
\end{equation}

\begin{equation}
\tilde{\sigma_t} \Equaldef \sqrt{\frac{\alpha^2_t\sigma_t^2}{\alpha^2_t+\sigma_t^2}} \ \ \text{and} \ \ d_t \Equaldef\frac{\sigma_t^2}{\alpha^2_t+\sigma_t^2}.
\end{equation}

\vspace{5pt}

\subsection{Best response}

The best response (BR) to a strategy profile where the policies are affine is given by 
\begin{equation}
\sum_{j\in\mathcal{N}_i} \pi_{ij}(y^{(i)}) \stackrel{a_i^*=1}{\geq} \frac{|\mathcal{N}_i|}{2} + \frac{|\mathcal{N}_i|^2}{2N}\left(d_1y^{(i)}_1-d_2y^{(i)}_2\right).
\end{equation}

We consider the case when the system is comprised of homogeneous agents using identical policies. This assumption is justified in the case of \textit{regular} networks. 
However, the policies constituting a BNE are nonidentical when the graph is irregular.

\subsection{Regular graphs and symmetric strategy profiles}

In a regular graph, every node has the same degree. Therefore, $|\mathcal{N}_i| = K$, $i\in[N]$. A homogeneous strategy profile is one where every agent $j\in\mathcal{N}_i$ uses the same affine policy, i.e.,
\begin{equation}
\tau^{(j)} = \tau \ \ \text{and} \ \ a^{(j)} = a, \ \ j\in \mathcal{N}_i.
\end{equation} 
The belief $\pi_{ij}(y_i)$ is given by:
\begin{equation}
\pi_{ij} (y^{(i)}) = \mathbf{E}\left[\Phi\left( \frac{\tau - d_1a_1y^{(i)}_1 - d_2a_2y^{(i)}_2}{\sqrt{(a_1\alpha_1)^2+(a_2\alpha_2)^2}}
 - W\right)\right],
\end{equation}
with 
\begin{equation}
W\sim\mathcal{N}\left(0, \frac{(a_1\tilde{\sigma}_1)^2+(a_2\tilde{\sigma}_2)^2}{(a_1\alpha_1)^2+(a_2\alpha_2)^2}\right).
\end{equation}
Therefore, the BR to a homogeneous strategy profile is given by
\begin{equation}
\pi_{ij} (y^{(i)}) \stackrel{a_i^*=1}{\geq}\frac{1}{2}+ \frac{\rho}{2}\left(d_1y^{(i)}_1-d_2y^{(i)}_2\right),
\end{equation}
where $\rho \Equaldef K/N$ is the \textit{density} of the regular graph connecting the agents.

\vspace{5pt}


\section{Diffuse Gaussian prior distribution}

Consider the case when $\sigma_t^2\rightarrow \infty$, $t\in\{1,2\}$, which corresponds to an assumption widely used in the economics literature known as an uniformly distributed random variable supported on the entire real line. Instead, we refer to that case as the \textit{diffuse Gaussian prior} distribution. When $\sigma_t^2\rightarrow \infty$, $t\in\{1,2\}$, we have
\begin{equation}
\tilde{\sigma}_t^2 \rightarrow \alpha_t^2, \ \ \text{and} \ \ d_t \rightarrow 1, \ \ t\in\{1,2\},
\end{equation}
and $W_j \sim \mathcal{N}(0,1), \ \ j\in[N].$

\vspace{5pt}

\begin{lemma}\label{lem:identity}
If $W\sim\mathcal{N}(0,1)$, then 
\begin{equation}
\mathbf{E}\big[ \Phi(W)\big] =\frac{1}{2}
\end{equation}
\end{lemma}

\vspace{5pt}

\begin{proof}
Using polar coordinates, we have
\begin{IEEEeqnarray}{rCl}
\mathbf{E}\big[ \Phi(W)\big] & = & \frac{1}{2\pi}\int_{-\infty}^{\infty}\int_{-\infty}^{w}e^{\left(-\frac{w^2+\xi^2}{2}\right)}\mathrm{d}w\mathrm{d}\xi \\
& = & \frac{1}{2\pi}\int_{0}^{\infty}\int_{\frac{\pi}{4}}^{\frac{5\pi}{4}} re^{-\frac{r^2}{2}}\mathrm{d}\phi\mathrm{d}r = \frac{1}{2}.
\end{IEEEeqnarray}
\end{proof}

\vspace{5pt}

\begin{proof}\textit{(Proof of Theorem 1)}
Let $a_1=1$, $a_2=-1$ and $\tau=0$. Define the function $\mathcal{F}:\mathbb{R}^2\rightarrow \mathbb{R}$ as:
\begin{multline} \label{eq:Sigmoidal}
\mathcal{F}(y) \Equaldef \mathbf{E}\left[ \Phi\left(\frac{y_2-y_1}{\sqrt{\alpha_1^2+\alpha_2^2}}-W\right)\right] -\frac{1}{2} - \frac{\rho}{2}(y_2-y_1).
\end{multline}
Therefore, the best response when every opponent uses the policy in \cref{eq:min_policy} is given by
\begin{equation}\label{eq:BR}
\mathrm{BR}_i(y) = \begin{cases}
1, \ \ \text{if} \ \ \mathcal{F}(y) \geq 0 \\
2, \ \ \text{otherwise.}
\end{cases}
\end{equation}

Let $\mathcal{L}\Equaldef\{y \in \mathbb{R}^2\mid y_1 = y_2\}$. We have
\begin{equation}
\mathcal{F}(y) = \mathbf{E} \big[ \Phi(W)\big] - \frac{1}{2} \stackrel{(a)}{=} 0, \ \  y \in \mathcal{L}.
\end{equation}
where $(a)$ follows from \cref{lem:identity}.

For the second part of our argument we need to compute the gradient of $\mathcal{F}$, which is given by:
\begin{equation}
\nabla \mathcal{F}(y) = \underbrace{\Bigg(\mathbf{E}\Bigg[\frac{e^{-\frac{1}{2}\big(\frac{y_2-y_1}{\sqrt{\alpha_1^2+\alpha_2^2}}-W\big)^2}}{\sqrt{2\pi(\alpha_1^2+\alpha_2^2)}}\Bigg] + \frac{\rho}{2} \Bigg)}_{\Equaldef \delta(y)>0} \begin{bmatrix} -1 \\ +1\end{bmatrix}.
\end{equation}

Using the first order approximation of $\mathcal{F}$ around a point $\bar{y} \in \mathcal{L}$, we have:
\begin{IEEEeqnarray}{rCl}
\mathcal{F}(y) &\approx& \mathcal{F}(\bar{y}) + \nabla \mathcal{F}(\bar{y})\begin{bmatrix} y_1-\bar{y}_1 \\ y_2-\bar{y}_2\end{bmatrix} \\ & = &  \mathcal{F}(\bar{y}) + \delta(\bar{y})  \big(y_2-\bar{y}_2-y_1+\bar{y}_1\big)
\end{IEEEeqnarray}
Therefore, because $\delta(\bar{y})>0$, if
\begin{equation}
 y_1-\bar{y}_1< y_2-\bar{y}_2,
\end{equation}
then 
\begin{equation}
\mathcal{F}(y) > \mathcal{F}(\bar{y}).
\end{equation}
However, the function $\mathcal{F}$ remains constant  if
\begin{equation}
y_2-\bar{y}_2 = y_1-\bar{y}_1.
\end{equation}


Since for any point $\bar{y}\in \mathcal{L}$, there exists a sufficiently small neighborhood around it such that for all $y$ such that  $y_1-\bar{y}_1< y_2-\bar{y}_2$, we have $\mathcal{F}(y) > \mathcal{F}(\bar{y})$. Pick a point in this neighborhood and call it $\bar{\bar{y}}$, and repeat this argument to obtain a new neighborhood over which the function is increasing for all $y_1-\bar{\bar{y}}_1< y_2-\bar{\bar{y}}_2$, and so on. Finally, we conclude that $\mathcal{F}(y)<0$ for all points in $\{y \in \mathbb{R}^2 \mid y_1 \geq y_2 \}$.
\end{proof}

\begin{figure}[t!]
    \centering
    \includegraphics[width=\columnwidth]{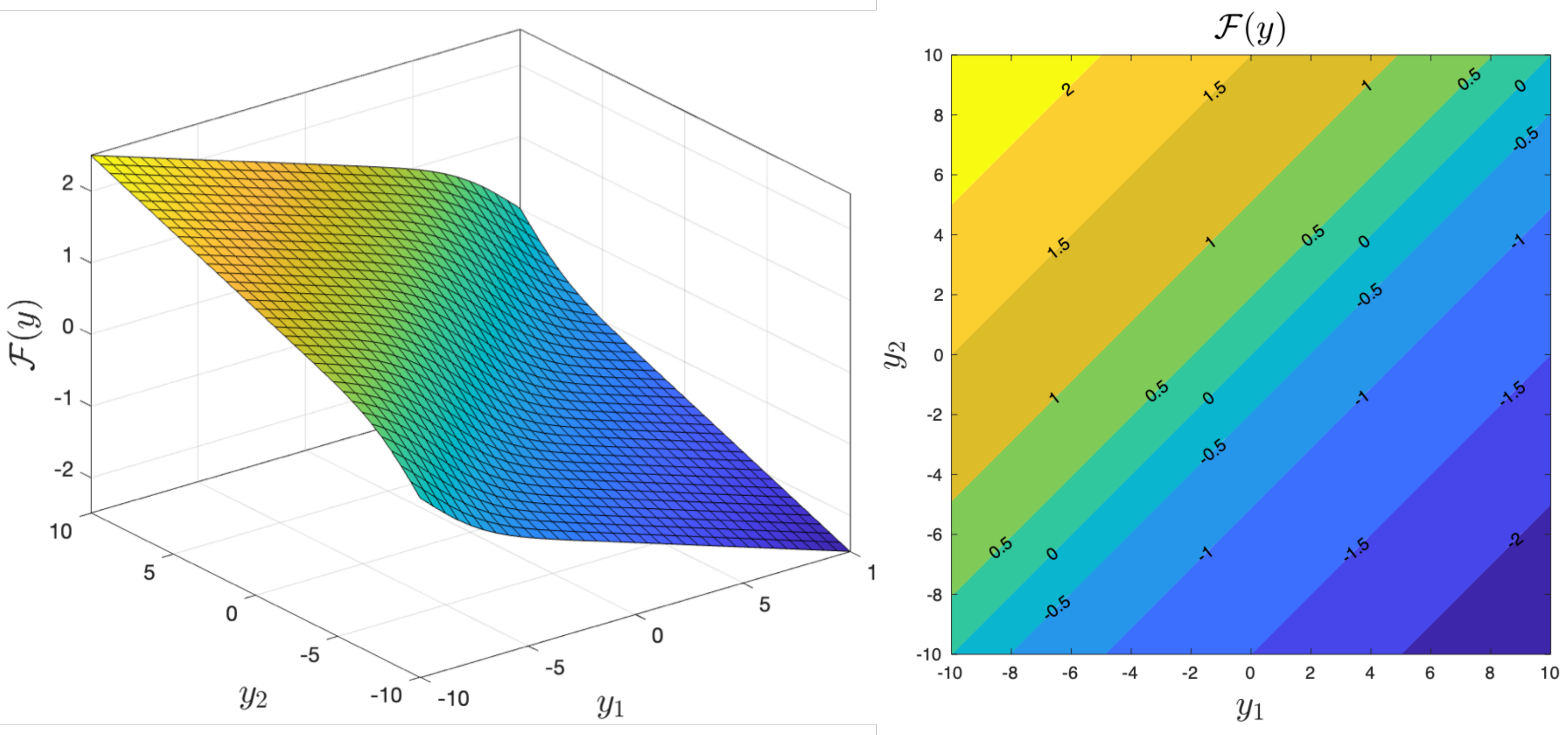}
    \caption{Sigmoidal function $\mathcal{F}$ in \cref{eq:Sigmoidal}, which defines the best-response policy in \cref{eq:BR}.}
    \label{fig:sigmoidal_BR}
\end{figure}

\vspace{5pt}

\section{Non-diffuse Gaussian prior distribution}

In the non-asymptotic regime, we show that the BR to an homogeneous affine strategy profile is characterized by the existence of a switching curve $g$ such that:
\begin{equation}
\mathrm{BR}_i(y) = \begin{cases}
1, & \ \ y_1 \leq g(y_2) \\
2, & \ \ \text{otherwise}. 
\end{cases}
\end{equation}

\begin{theorem}\label{thm:increasing_switching}
If $a_1>0$ and $a_2<0$, the best-response switching curve $g(\cdot)$ is monotone increasing.
\end{theorem}

\begin{proof}
Define $\mathcal{G}:\mathbb{R}^2\rightarrow \mathbb{R}$, such that:
\begin{multline}\label{eq:switching_curve1}
\mathcal{G}(\xi,y_2) \Equaldef \mathbf{E}\left[ \Phi\left(\frac{\tau-a_1d_1\xi - a_2d_2y_2}{\sqrt{a_1^2\alpha_1^2+a_2^2\alpha_2^2}}-W\right)\right] -\frac{1}{2} \\ -\frac{\rho}{2}\left(d_1\xi-d_2y_2\right).
\end{multline}

Define $g(y_2)$ as the set:
\begin{equation}\label{eq:switching_curve2}
g(y) \Equaldef \{\xi \in \mathbb{R} \mid \mathcal{G}(\xi,y) =0 \}.
\end{equation}

First we show that if $a_1>0$, the set $g(y)$ is a singleton for every $y_2\in\mathbb{R}$. For a fixed $y_2\in \mathbb{R},$ we will show that the solution of $\mathcal{G}(\xi,y_2) =0$ exists and is unique.
Let $\mathcal{G}_1(\xi,y_2)$ and $\mathcal{G}_2(\xi,y_2)$ be defined as
\begin{equation}
\mathcal{G}_1(\xi,y_2) \Equaldef \mathbf{E}\left[ \Phi\left(\frac{\tau-a_1d_1\xi - a_2d_2y_2}{\sqrt{a_1^2\alpha_1^2+a_2^2\alpha_2^2}}-W\right)\right]
\end{equation} 
and
\begin{equation}
\mathcal{G}_2(\xi,y_2) \Equaldef \frac{1}{2}  +\frac{\rho}{2}\left(d_1\xi-d_2y_2\right).
\end{equation}
Taking the partial derivatives of $\mathcal{G}_1(\xi,y_2)$ and $\mathcal{G}_2(\xi,y_2)$ with respect to $\xi$ yields:
\begin{multline}
\frac{\partial}{\partial \xi}\mathcal{G}_1(\xi,y) = \frac{-a_1d_1}{\sqrt{2\pi(a_1^2\alpha_1^2+a_2^2\alpha_2^2)}} \times \\ \mathbf{E}\left[ \exp\left(-\frac{1}{2}\left(\frac{\tau-a_1d_1\xi - a_2d_2y_2}{\sqrt{a_1^2\alpha_1^2+a_2^2\alpha_2^2}}-W\right)^2\right)\right]
\end{multline} 
and
\begin{equation}
\frac{\partial}{\partial \xi}\mathcal{G}_2(\xi,y_2) = \frac{\rho }{2}d_1.
\end{equation}

Then, we conclude that if $a_1>0$, the function $\mathcal{G}_1$ is monotone decreasing while $\mathcal{G}_2$ is monotone increasing in $\xi$. Moreover, 
\begin{equation}
\lim_{\xi \rightarrow -\infty} \mathcal{G}_1(\xi,y_2) = 1 \ \ \text{and} \ \ \lim_{\xi \rightarrow +\infty} \mathcal{G}_1(\xi,y_2) = 0,
\end{equation}
while
\begin{equation}
\lim_{\xi \rightarrow -\infty} \mathcal{G}_2(\xi,y_2) = -\infty \ \ \text{and} \ \ \lim_{\xi \rightarrow +\infty} \mathcal{G}_2(\xi,y_2) = +\infty.
\end{equation}

Therefore, there always exists a single point $\xi^\star$ such that $\mathcal{G}_1(\xi,y_2)=\mathcal{G}_2(\xi,y_2)$.

We now show that if $a_1>0$ and $a_2<0$, the function $g(\cdot)$ is monotone increasing. Because the function $g(y)$ is implicitly defined as the unique solution of $\mathcal{G}(\xi,y) =0$, the derivative of $g(y)$ must be computed using implicit differentiation:
\begin{equation}
g'(y)=\frac{\mathrm{d}\xi}{\mathrm{d}y} = -\frac{\frac{\partial}{\partial y}\mathcal{G}(\xi,y)}{\frac{\partial}{\partial \xi}\mathcal{G}(\xi,y)}.
\end{equation}
One can show that if $a_1>0$ and $a_2<0$, we have
\begin{equation}
\frac{\partial}{\partial \xi}\mathcal{G}(\xi,y)<0 \ \ \text{and} \ \ \frac{\partial}{\partial \xi}\mathcal{G}(\xi,y)>0.
\end{equation}
Therefore, $g'(y)>0$.
\end{proof}

\vspace{5pt}

\begin{example}[Non-linearity of the best-response policy] Consider the following parameters for a system with $N=10$ agents over the regular network in \cref{fig:system} (i.e., $K=4$): $\sigma_1^2 = 1$, $\sigma_2^2 = 2$ and $\alpha_1^2 = \alpha_2^2 = 1$. The best-response policy's switching curve to an homogeneous strategy profile where $a_1 =1$, $a_2 = -2$ and $\tau=0$ is shown in \cref{fig:nonlinear}. As we can observe, the switching curve $g$ is increasing as implied by \cref{thm:increasing_switching}. However, it is non-linear. \Cref{fig:nonlinear} also shows two linear functions: one that approximates $g$ well in the interval $[-20,+20]$ and another that approximates it well when $y^{(i)}_2\rightarrow \pm \infty$. 
\end{example}

\begin{figure}[t!]
    \centering
    \includegraphics[width=\columnwidth]{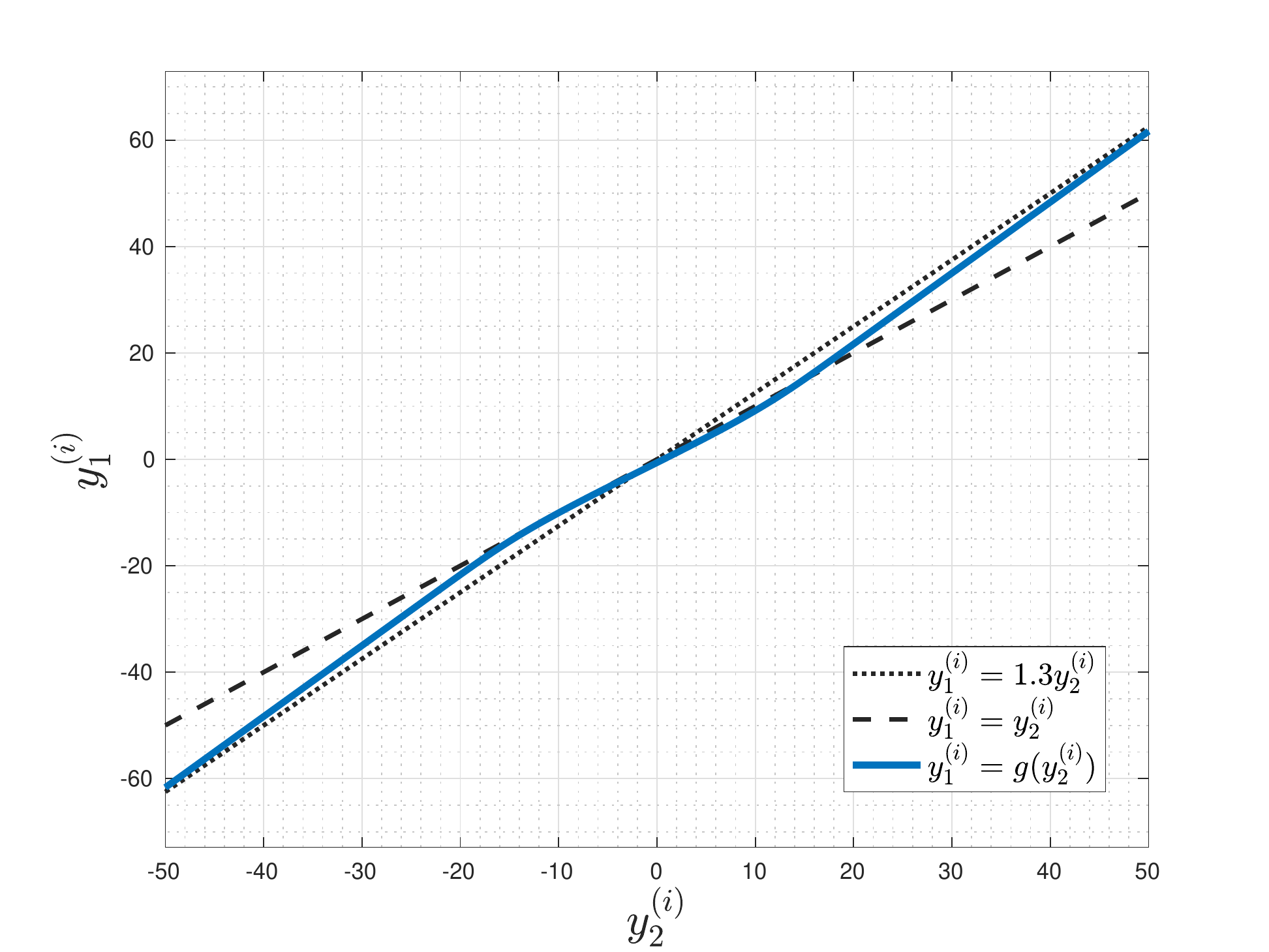}
    \caption{Best response policy's switching curve.}
    \label{fig:nonlinear}
\end{figure}

\subsection{Affine Bayesian Nash-Equilibrium}

Since the BR to a strategy profile consisting of homogeneous affine policies may be non-linear, we propose an iterative algorithm to compute an affine homogeneous BNE. In our algorithm, we use a special operator that projects the switching function $g(\cdot)$ onto the space of affine functions. Let the projection operator $\mathcal{P}$ be defined as follows: 
\begin{equation}
\mathcal{P}(g) \Equaldef \arg\min_{a_2,\tau} \ \mathbf{E} \Big[ \big( a_2Y_2 + \tau -g(Y_2) \big)^2 \Big].
\end{equation}

There are two reasons behind computing the affine approximation using this operator: 1. the points that are more likely to occur receive more weight than the unlikely ones, resulting in lower approximation errors; 2. it allows data-driven algorithms based on empirical risk minimization \cite{Shalev:2014}.

To compute the projection we need to evaluate the function $g$ over the entire real line, which can be a computationally prohibitive. Instead, we replace the projection with its approximate version. Let $\mathcal{D}_M=\{y_2^{(i)}(m)\}_{m=1}^M$ be a set of $M$ i.i.d. samples drawn from $f_{Y_2}$, define $\hat{\mathcal{P}}$ as 
\begin{equation}\label{eq:approximate_proj}
\hat{\mathcal{P}}(g) \Equaldef \arg\min_{a_2,\tau} \ \frac{1}{M} \sum_{m=1}^M \Big( a_2y_2(m) + \tau -g\big(y_2(m)\big) \Big)^2 .
\end{equation}

Algorithm 1 can be used to obtain symmetric affine coordination policies using the data-driven approximation to the projection operator stated above.

\begin{algorithm}[t]
    \caption{Affine Bayesian Nash Equilibrium}
    \label{alg: I}
    \begin{algorithmic}[1]
        \REQUIRE  Prior variances $(\sigma^2_1,\sigma^2_2)$, noise variances $(\alpha^2_1,\alpha^2_2)$, and number of agents $N$, dataset size $M$

        \ENSURE  Find  $a_2^{\star}$ and $\tau^{\star}$
        \STATE
        Initialize 
          $a_2<0$ and $\tau$, $\mathcal{D}_M \gets \{y_2(m)\}_{m=1}^M$
        \REPEAT

        \STATE Compute $\{g(y_2(m))\}_{m=1}^M$ using \cref{eq:switching_curve1,eq:switching_curve2}

        \STATE $(a_2,\tau) \gets \hat{\mathcal{P}}(g)$ using \cref{eq:approximate_proj}

        \UNTIL $(a_2,\tau)$ converge
    \end{algorithmic}
\end{algorithm}

\section{Numerical results}

Algorithm 1 provides a practical way of implementing coordination policies guided by theoretical optimality principles. Consider the following example where: $\sigma_1^2=2$, $\sigma_2^2=1$, $\alpha_1^2=1$ and $\alpha_2^2=1$. Consider a system with $N=10$ and regular graphs with different densities $\rho$. Running Algorithm 1, we obtained \cref{tab:policies}, where we notice that as the graph becomes more sparse (lower values of $\rho$), the threshold $\tau^\star$ increases to the point of becoming dominant term in the switching function of the affine BNE policies. However, as the graph becomes more dense, the threshold becomes negligible, and the BNE coordination policy is dominated by the linear term. An interesting feature of this problem is that for regular networks, the numerical results only depend on the density of the graph, and not on the number of agents $N$. It is not clear at this point what is the role of $N$ in this problem. However, \cref{tab:policies} implies that even though we are allowing for a more general class of policies for a multi-task assignment problem, if $\rho\rightarrow 0$, e.g. $K=\log N$, the optimal policy converges to a simple threshold policy.

\begin{table}[t!]
\centering
\caption{Affine BNE switching functions for networks with density $\rho$}
\label{tab:policies}
\begin{tabular}{|c| c| c |c|}
\hline 
$\rho$ & $a_2^\star$ & $\tau^\star$ & residual error \\
\hline \hline
 $0.1$ & $-1.3259$ & $+19.979$ & $0.0372$ \\ \hline
 $0.2$ & $-1.2433$ & $+9.2244$  & $0.1046$ \\ \hline
$ 0.3$ & $-1.2056$ & $+3.4724$  & $0.0278$ \\ \hline
 $0.4$ & $-1.2070$ &  $-3.69\times 10^{-4}$      & $ 0.0022$ \\ \hline
 $0.5 $& $-1.2156$ & $-1.75\times 10^{-4}$  & $0.0026$  \\ \hline
 $0.6$ & $-1.2231$& $-2.64\times 10^{-5}$& $0.0027$\\ \hline
  $0.7$ & $-1.2297$& $-9.39\times 10^{-5}$ & $0.0019$\\ \hline
 $0.8 $& $-1.2356$ & $-9.94\times 10^{-5}$ & $0.0032$\\ \hline
$ 0.9$ & $-1.2408$ & $+5.85\times 10^{-5}$& $0.0022$\\ \hline
\end{tabular}

\end{table}




\vspace{5pt}


\vspace{5pt}


\vspace{3pt}




\section{Conclusions and Future work}

We introduced a framework of task allocation in regular robotic networks based on global games. Our formulation combines networked interactions and multiple fundamentals, and admits a simple linear BNE when the prior distributions are diffuse. When the priors have finite variance, we showed that the the switching function is nonlinear, but can be approximated reasonably well by affine policies. Using an expected $L_2$ distance, we provide an efficient algorithm that can be used to find approximate BNE in the class of affine policies. The next step in this work is to allow the network to be irregular and to study the fundamental limits of coordination as a function of the signal to noise ratio in the private signals.


\bibliographystyle{IEEEtran}

\bibliography{IEEEabrv,MTGG}

\end{document}